\documentclass[11pt]{article}
\usepackage{color}
\usepackage{colortbl}
\usepackage{multicol}
\usepackage{booktabs, tabularx, makecell, multirow,caption}
\usepackage{subcaption}
\usepackage{amsmath}
\usepackage{amsthm}
\usepackage{balance}
\usepackage[ruled, noline, linesnumbered, commentsnumbered, noend]{algorithm2e}
\usepackage{thm-restate}
\usepackage[end]{algpseudocode}
\usepackage{xspace}
\usepackage[symbol]{footmisc}
\usepackage{booktabs, tabularx, makecell, multirow}
\usepackage{thmtools}
\usepackage[shortlabels]{enumitem}
\usepackage{wrapfig}
\usepackage{graphicx}
\usepackage{subcaption}
\usepackage{todonotes}
\usepackage[sort]{cite}
\usepackage{tikz}
\usepackage{subcaption} 
\usepackage{wrapfig} 
\usepackage{hyperref}
\usepackage[capitalize,nameinlink]{cleveref}
\usepackage{appendix}
\usepackage{xparse}
\usepackage{etoolbox}
\usepackage{amsfonts}

\theoremstyle{plain}
\newtheorem{theorem}{Theorem}[section]
\newtheorem{lemma}[theorem]{Lemma}
\newtheorem{proposition}[theorem]{Proposition}

\newtheorem*{theorem*}{Theorem}
\newtheorem*{lemma*}{Lemma}

\theoremstyle{definition}
\newtheorem{definition}[theorem]{Definition}

\theoremstyle{remark}

\usepackage{amsmath, amsthm}
\usepackage[capitalize,nameinlink]{cleveref}
\usepackage{appendix}
\usepackage{xparse}

\title{Adaptive BSTs for Single-Source and All-to-All Requests: Algorithms and Lower Bounds} 
\author{Maryam Shiran}
\date{}

\newcommand{\appsymb}{\textsuperscript{\hyperref[app: omitted proofs]{\ensuremath{\star}}}}

\newcommand{\appref}[1]{\appsymb\label{#1}}

\newcommand{\appendixproofs}{} 
\newcommand{\appendixproof}[2]{%
  \gappto{\appendixproofs}{%
    \begin{proof}[Proof of Lemma~\ref{#1}]
    #2
    \end{proof}
  }%
}

\begin{document}

\maketitle

\begin{abstract}
Adaptive binary search trees are a fundamental data structure for organizing hierarchical information. Their ability to dynamically adjust to access patterns makes them particularly valuable for building responsive and efficient networked and distributed systems.

We present a unified framework for adaptive binary search trees with fixed restructuring cost, analyzed under two models: the single-source model, where the cost of querying a node is proportional to its distance from a fixed source, and the all-to-all model, where the cost of serving a request depends on the distance between the source and destination nodes. We propose an offline algorithm for the single-source model and extend it to the all-to-all model. For both models, we prove upper bounds on the cost incurred by our algorithms. Furthermore, we show the existence of input sequences for which any offline algorithm must incur a cost comparable to ours.

In the online setting, we develop a general mathematical framework for deterministic online adaptive binary search trees and propose a deterministic online strategy for the single-source case, which naturally extends to the all-to-all model. We also establish lower bounds on the competitive ratio of any deterministic online algorithm, highlighting fundamental limitations of online adaptivity.
\end{abstract}

\vspace{1em}

\section{Introduction}
\label{sec: introduction}
Tree data structures have long served as a foundational abstraction in computer science, enabling efficient representation and manipulation of hierarchical data~\cite{cormen2022introduction}. Among these, \emph{binary search trees} (BSTs) are particularly prominent due to their ability to support dynamic set operations efficiently in both theoretical and practical settings~\cite{bentley1975multidimensional, cormen2022introduction}. Over the years, a wide range of BST variants have been developed to improve performance, ensuring balance, adaptivity, or theoretical optimality under different assumptions~\cite{sleator1985self, demaine2007dynamic, galperin1993scapegoat, seidel1996randomized}.

Complementing these static approaches, a significant body of research has focused on \emph{self-adjusting} or \emph{adaptive} data structures, which reorganize themselves dynamically in response to access patterns~\cite{sleator1985self, tarjan1985amortized, slastin2023efficient}. These structures aim to maintain good performance without prior knowledge of the input distribution, leveraging temporal or spatial locality through online reconfiguration. A seminal example is the splay tree~\cite{sleator1985self}, which moves the accessed node to the root via rotations and achieves amortized logarithmic performance. This idea has inspired a variety of follow-up works proposing structures that adapt through self-organization, randomization, or access-sensitive strategies~\cite{wang2006multi, cole2000dynamic, slastin2023efficient, chalermsook2015greedy, demaine2007dynamic}. 

Building on this foundation, recent research has extended the use of adaptive trees to distributed and networking contexts~\cite{schmid2015splaynet, avin2019toward, feder2022lazy, pourdamghani2023seedtree}. For instance, \emph{SplayNet} adapts the idea of splay trees to distributed systems, allowing network topologies to self-adjust based on observed communication patterns~\cite{schmid2015splaynet}. By reorganizing the tree to bring frequently communicating nodes closer together, SplayNet reduces routing costs over time. 

\textbf{In this work}, we examine a unified model that builds upon splay trees and SplayNet. Specifically, we consider \emph{adaptive binary search trees} used for either node access (as in classical BSTs) or pairwise communication requests (as in network settings). In both cases, the BST serves a request sequence by optionally changing its structure at a cost $C(n)$, where $n$ is the number of nodes. The cost of serving a request is defined analogously to splay trees or SplayNet: either the depth of the accessed node (single-source model), or the length of the shortest path between two nodes (all-to-all model). We now summarize our main theoretical results for the unified adaptive‑BST model under cost $C(n)$:
\begin{itemize}
    \item \textbf{Offline upper bounds.} In Section~\ref{sec: lowerbound}, we provide an offline algorithm for the single-source model and we extend it for the all-to-all model. For any arbitrary request sequence $\sigma = (\sigma_i)_{i=1}^m$ (single-source model) or $\sigma = ((\sigma_{s_i}, \sigma_{d_i}))_{i=1}^m$ (all-to-all model), we achieve cost upper bounds of
    \[
    2m \log_2\bigl(C(n) + 2.2\bigr) \quad \text{and} \quad 4m \log_2\bigl(C(n)\bigr) + 3.9,
    \]
    respectively.
    
    \item \textbf{Offline lower bounds.} In Section~\ref{sec: lowerbound}, we show that there exist sequences for which any offline algorithm must incur costs of at least
    \[
    m \log_2\bigl(C(n)\bigr) \quad \text{and} \quad \tfrac{1}{4} m \log_2\bigl(C(n)\bigr),
    \]
    for the single-source and all-to-all models, respectively.
    
    \item \textbf{Online algorithm framework.} In Section~\ref{sec: alg}, we develop a mathematical framework for analyzing online algorithms by decomposing them into fundamental components, under the assumption that the query sequence is independent of the tree state. We also present our proposed algorithm along with its design motivation.
    
    \item \textbf{Competitive ratios' lower bounds.} In Section~\ref{sec: alg}, we prove that no online algorithm can achieve a competitive ratio better than
    \[
    \frac{\log_2(n)}{2 \log_2(C(n) + 2.2)} \quad \text{(single-source)} \quad \text{and} \quad \tfrac{1}{4} \log_2 n \quad \text{(all-to-all)}.
    \]

\end{itemize}

\section{Models \& Definitions}
\label{sec: model}

\begin{figure}[ht]
    \centering
    \begin{subfigure}[b]{0.48\textwidth}
        \centering
        \resizebox{0.95\linewidth}{!}{%
            \begin{tikzpicture}[level distance=1.2cm, 
                                level 1/.style={sibling distance=3cm},
                                level 2/.style={sibling distance=1.8cm},
                                level 3/.style={sibling distance=1.2cm},
                                every node/.style={circle, draw, minimum size=6mm, inner sep=1pt}]
                \node (root) {5}
                    child {node (A) {3}
                        child {node (B) {2}
                            child {node (C) {1}}
                            child [missing] {}
                        }
                        child {node {4}}
                    }
                    child {node {7}
                        child {node {6}}
                        child {node {8}}
                    };
                \draw[blue, line width=1.5pt, rounded corners=10pt] 
                    (root) -- (A) -- (B) -- (C);
            \end{tikzpicture}
        }
        \caption{Request $\sigma_i=1$ in the single-source model.}
        \label{fig:tree-model}
    \end{subfigure}
    \hfill
    \begin{subfigure}[b]{0.48\textwidth}
        \centering
        \resizebox{0.95\linewidth}{!}{%
            \begin{tikzpicture}[level distance=1.2cm, 
                                level 1/.style={sibling distance=3cm},
                                level 2/.style={sibling distance=1.8cm},
                                level 3/.style={sibling distance=1.2cm},
                                every node/.style={circle, draw, minimum size=6mm, inner sep=1pt}]
                \node (root) {5}
                    child {node (A) {3}
                        child {node (B) {2}
                            child {node (C) {1}}
                            child [missing] {}
                        }
                        child {node (D) {4}}
                    }
                    child {node {7}
                        child {node {6}}
                        child {node {8}}
                    };
                \draw[blue, line width=1.5pt, rounded corners=10pt] 
                    (C) -- (B) -- (A) -- (D);
            \end{tikzpicture}
        }
        \caption{Request $(\sigma_{s_i}=1, \sigma_{d_i}=4)$ in the all-to-all model.}
        \label{fig:network-model}
    \end{subfigure}

    \caption{Binary search tree as (a) a routing structure in the single-source model, and (b) a communication network in the all-to-all model. The blue path indicates the cost-incurring traversal.}
    \label{fig:combined-models}
\end{figure}
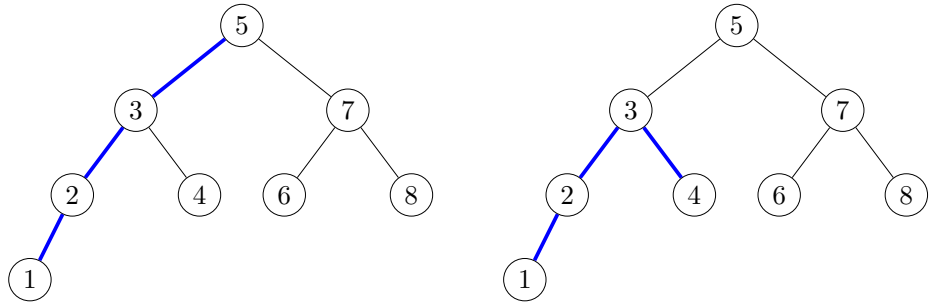
This paper explores two models: first, a single-source model in which all requests come from a single-source; and second, an all-to-all model a network model, in which request happen between pair of nodes. In Subsections~\ref{subsec: Tree-Based-model} and~\ref{subsec: Network-Based-model}, we provide precise definitions of each.

\subsection{Single-source model}
\label{subsec: Tree-Based-model}

The single-source model consists of a structure with \( n \) nodes, represented by the set \( V = \{1, \dots, n\} \), and a finite request sequence $\sigma = (\sigma_i)_{i=1}^m$, where each request \( \sigma_t \in V \) corresponds to a query for node \( u = \sigma_t \). At each stage, the topology \( G_i \) is selected from a predefined family of tree topologies \( \mathcal{G} \), where each \( G \in \mathcal{G} \) is a binary search tree \( G = (V, E) \), consisting of the node set \( V \) and the edge set \( E \). 
Each request \(\sigma_t\) is served within the current tree topology \( G_{t-1} \), with routing originating from the root node. 
The routing cost for accessing node \( u \) in topology \( G_i \), denoted by \( d_{G_{i}}(u) \), is defined as the length of the shortest path from the root to \( u=\sigma_t \) in \(G_{t}\). Figure~\ref{fig:tree-model} illustrates an example of a single-source model.
The reconfiguration cost incurred when modifying the tree topology between consecutive requests is defined as

\[
\text{Adjustment Cost}(G_i, G_{i-1}) =
\begin{cases}
    0, & \text{if } G_i = G_{i-1}, \\
    C(n), & \text{otherwise},
\end{cases}
\]

We assume $C(n)$ to be a function of \(n\), satisfying $C(n)\le n$. The total cost of processing the request sequence \( \sigma \) is given by the sum of routing costs for each request and the reconfiguration costs incurred due to topology adjustments, expressed as
\[
\text{Total Cost}(G_0, \sigma) = \sum_{i=1}^{m} \left(d_{G_{i-1}}(\sigma_i) + \text{Adjustment Cost}(G_{i-1}, G_i)\right).
\]
\begin{proposition}[Single-Source Total Cost Expression]
\label{prop:stage-cost}
Let \( k \) denote the number of times the model changes its topology while processing the request sequence, resulting in \( k + 1 \) distinct stages.  
Let \( \sigma^i \) be the subsequence of communication requests served during stage \( i \), for \( i = 1, 2, \dots, k + 1 \), and let \( x_i = |\sigma^i| \) be its length.  
Let \( c_{\sigma^i} \) denote the total service cost incurred in stage \( i \) under the corresponding topology.  
Then, the total cost of the algorithm can be expressed as:
\(
\text{Total Cost} = k \cdot C(n) + \sum_{i=1}^{k+1} c_{\sigma^i},
\)
where \( C(n) \) denotes the cost of changing the topology, assumed to depend only on the number of nodes \( n \). We also denote $H_i$ as the entropy of the subsequence  \( \sigma^i \).
\end{proposition}

\subsubsection*{Problem settings}
\label{subsub: probset}
We consider two problem settings:  
\begin{itemize}
    \item \textbf{Offline:} The entire request sequence \( \sigma \) is known in advance. The goal is to construct a sequence of tree topologies \( G_i \in \mathcal{G} \) that minimizes the total cost.  
    \item \textbf{Online:} Requests arrive over time, and the tree topology can be adjusted after each request, starting from an initial topology \( G_0 \in \mathcal{G} \). The objective is to update \( G_i \) at each step \( i = 1, \dots, m \) to minimize the total cost.  
\end{itemize}

\subsection{All-to-all model}
\label{subsec: Network-Based-model}

The all-to-all model consists of a structure with  \( n \) nodes, denoted by the set \( V = \{1, \dots, n\} \), and a finite set of communication requests $\sigma = ((\sigma_{s_i}, \sigma_{d_i}))_{i=1}^m$, where each pair \( (\sigma_{s_i}, \sigma_{d_i}) \) represents a communication request from source node \( s_i \) to destination node \( d_i \).
The network topology, denoted as \( G_i \), is selected from a predefined family of topologies \( \mathcal{G} \), where each \( G \in \mathcal{G} \) is a binary search tree \( G = (V, E) \), consisting of the node set \( V \) and edge set \( E \). Same as the work \emph{SplayNet}~\cite{schmid2015splaynet} we assume the routing cost for serving a communication request between nodes \( u=\sigma_{s_t} \) and \( v=\sigma_{d_t} \), denoted \( d_{G_{t-1}}(u, v) \), is defined as the shortest path length between the source \( u \) and destination \( v \) in the current topology \( G_{t-1} \). Figure~\ref{fig:network-model} illustrates an example of an all-to-all model.
The reconfiguration cost is again defined as the Adjustment cost in section ~\ref{subsec: Tree-Based-model}, enabling us to express the total cost as
\[
\text{Total Cost}(G_0, \sigma) = \sum_{i=1}^{m} \left( d_{G_{i-1}}(\sigma_{s_i}, \sigma_{d_i}) + \text{Adjustment Cost}(G_{i-1}, G_i)\right).
\]
\begin{proposition}[All-to-All Total Cost Expression]
\label{prop:stage-cost-all-to-all}
Let $\sigma = ((\sigma_{s_i}, \sigma_{d_i}))_{i=1}^m$ be a sequence of pairwise communication requests.
Let $k$ denote the number of times the model changes its topology while serving the sequence, resulting in $k + 1$ distinct stages.
Let $\sigma^i \subseteq \sigma$ be the subsequence of communication requests handled in stage $i$, for $i = 1, 2, \dots, k + 1$, and let $x_i = |\sigma^i|$ be its length.
Let $c_{\sigma^i}$ denote the total routing cost for serving all requests in $\sigma^i$ under the corresponding topology.
Then, the total cost incurred by the algorithm is:
$
\text{Total Cost} = k \cdot C(n) + \sum_{i=1}^{k+1} c_{\sigma^i},
$
where $C(n)$ is the cost of changing the topology, depending only on the number of nodes $n$.
\end{proposition}
\subsubsection*{Problem settings}
\label{subsec: tree}
Like the problem settings defined in Subsubsection~\ref{subsub: probset}, both the offline and online settings are defined here as well. The topology remains the same; the only differences are that the sequence now consists of communication requests—i.e., source-destination pairs—instead of single nodes being queried, and that the cost of answering a query is measured by the shortest path length between the two nodes, rather than by the depth of a node.
\subsection{Competitive Ratio}
\label{sub: comp}
To assess the performance of online algorithms, we compare them to an optimal offline algorithm with full knowledge of the input. This comparison is captured by the competitive ratio, which quantifies the worst-case gap between the two.

\begin{definition}
The competitive ratio ~\cite{borodin2005online} is a measure of how well an online algorithm $ALG$ performs compared to an optimal offline algorithm  $OPT$ in terms of the cost each incurs. It is defined as
\(
\text{competitive ratio} = \max_{\sigma} \left( \frac{C_{\text{ALG}}(\sigma)}{C_{\text{OPT}}(\sigma)} \right).
\)
Where $C_{ALG}(\sigma)$ and $C_{OPT}$ are the costs each of $ALG$ and $OPT$ incurs for answering the sequence $\sigma$.
\end{definition}
We'll discuss lower bounds for the competitive ratio of any online algorithm, in both models, in subsection ~\ref{com: kol}.

\section{Analysis of offline algorithms}
\label{sec: lowerbound}
We analyze the performance of offline algorithms in both the single-source and all-to-all models. Specifically, we propose algorithms that achieve upper bounds of
\(
2m \log_2\bigl(C(n) + 2.2\bigr)
\quad \text{and} \quad
4m \log_2\bigl(C(n)\bigr) + 3.9,
\)
respectively, for arbitrary sequences
\(
\sigma = \left( \sigma_i \right)_{i=1}^m
\quad \text{and} \quad
\sigma = \left( (\sigma_{s_i}, \sigma_{d_i}) \right)_{i=1}^m.
\)
We also establish corresponding lower bounds, showing that there exist sequences for which the cost of any offline algorithm is at least
\[
m \log_2\bigl(C(n)\bigr)
\quad \text{and} \quad
\tfrac{1}{4} m \log_2\bigl(C(n)\bigr),
\]
in the single-source and all-to-all settings, respectively.
\subsection{Single-source model}
\label{sub_off_anal}
We begin with the single-source setting. We show that Algorithm~\ref{alg:offline} achieves a cost upper bounded by
\(
2m \log_2(C(n) + 2.2),
\)
and we prove the existence of sequences for which any offline algorithm must incur a cost of at least
\(
m \log_2(C(n)).
\)
\begin{algorithm}[h]
\caption{Offline Algorithm with Cost \(\leq 2m \log_2(C(n) + 2.2)\)}
\label{alg:offline}
\SetKwInOut{Input}{Input}
\SetKwInOut{Output}{Output}

\Input{
  $m$: Length of request sequence \\
  $C(n)$: Topology change cost \\
  $\sigma$: Request sequence
}
\Output{Partition of $\sigma$ and corresponding trees}

$k \gets \left\lceil \frac{m}{C(n) \ln 2} \right\rceil - 1$\;
$x \gets \left\lfloor \frac{m}{k+1} \right\rfloor$\;
$r \gets m \bmod (k+1)$\;

Partition $\sigma$ into $k+1$ contiguous blocks:\;
First $r$ blocks of size $x+1$, remaining blocks of size $x$\;

\ForEach{subsequence $\sigma^{(i)}$ in the partition\quad}{
  Build an optimal BST for $\sigma^{(i)}$ using Theorem~\ref{them: kurt-upper}\;
}
\Return partition and trees\;
\end{algorithm}

\begin{lemma}
\label{lem: off-given}
Given a single-source model, offline algorithm ~\ref{alg:offline} incurs at most cost \(2m \log_2(C(n) + 2.2)\) for any sequence \(\sigma = (\sigma_i)_{i=1}^m\).
\end{lemma}

\begin{proof}
Following the notion introduced in Proposition~\ref{prop:stage-cost}, Algorithm~\ref{alg:offline} proposes the total cost as $ kC(n) + \sum_{i=1}^{k+1} c_{\sigma^i}$, where \( k = \left\lceil \frac{m}{C(n) \ln(2)} \right\rceil - 1 \), and \( x_i \)'s are either 
\( \left\lceil \frac{m}{\left\lceil \frac{m}{C(n) \ln(2)} \right\rceil} \right\rceil \) or 
\( \left\lfloor \frac{m}{\left\lceil \frac{m}{C(n) \ln(2)} \right\rceil} \right\rfloor \), such that their sum equals \( m \). 
All of the \( x_i \)'s are bounded by \( C(n) \ln(2) + 1 \) because:
\[
x_i \leq \left\lceil \frac{m}{\left\lceil \frac{m}{C(n) \ln(2)} \right\rceil} \right\rceil 
\leq \frac{m}{\left\lceil \frac{m}{C(n) \ln(2)} \right\rceil} + 1 
\leq \frac{m}{\frac{m}{C(n) \ln(2)}} + 1 = C(n)\ln(2) + 1.
\]
We invoke Theorem~\ref{them: kurt-upper}, restated in Appendix~\ref {app: threstate} for completeness from the prior work \emph{Nearly Optimal Binary Search Trees}~\cite{mehlhorn1975nearly}.
Using the tree obtained from Theorem~\ref{them: kurt-upper} for each sub-sequence, it follows that
\[
\text{cost} = kC(n) + \sum_{i=1}^{k+1} c_{\sigma^i} \leq kC(n) + 2m + \sum_{i=1}^{k+1} \frac{x_i H_i}{1 - \log_2 (\sqrt{5} - 1)}.
\]
As previously assumed, \( C(n) \leq n \), and thus 
\(x_i\le  C(n) \ln(2) + 1 \leq n \). Each stage has \( d_i \) distinct elements where \( d_i \leq x_i \), so we have \(H_i \leq \log_2 (d_i) \leq \log_2 (x_i) \), and can write:

\begin{align*}
\text{cost} 
&\leq kC(n) + 2\left(m + \sum_{i=1}^{k+1} x_i \log_2 (x_i)\right) \\
&\leq \left(\frac{m}{C(n) \ln(2)}\right)C(n) + 2\left(m + m \log_2 (C(n) \ln(2) + 1)\right) \\
&\leq m \cdot \frac{1}{\ln(2)} + m\left(1 + \log_2 (C(n) \ln(2) + 1)\right) \\
&\leq 2m \cdot \log_2 (C(n) + 2.2).
\end{align*}

\end{proof}
Notably, this algorithm does not utilize access to the entire sequence for its computations; rather, it relies only on access to the next $ C(n)\ln(2) + 1$ elements at certain points. This property allows it to function as an online algorithm if a portion of the sequence can be accessed in advance.\\
Now, we'll show that there are sequences for which, any offline algorithm has to pay at least \(m \log_2 (C(n))\).
\begin{lemma}
\label{lem: low-Hole}
Given the special sequence 
\[
\sigma = \left( \sigma_i \right)_{i=1}^m 
\quad \text{where} \quad 
\sigma_i = \left( (i \mod n) + 1 \right),
\]
any offline algorithm pays a cost of at least 
\(
m \log_2 (C(n)).
\)
\end{lemma}

\begin{proof}
To prove this lemma, we follow the notion introduced in Proposition~\ref{prop:stage-cost} and establish two supporting lemmas as follows.
\begin{lemma}\appref{lem: lowes-on-ent}\footnote{The proofs of statements marked by \appsymb\ are deferred to \Cref{app: omitted proofs}.}
Following the notion introduced in Proposition~\ref{prop:stage-cost} and given the special sequence,
\(
\sigma = \left( \sigma_i \right)_{i=1}^m \quad \text{where} \quad \sigma_i = \left( (i \mod n) + 1 \right)
\)
for some natural numbers \(m, n \), if \(x_i > n\), then \(H_i \ge \frac{\log_2 n}{2}\). On the other hand, if \(x_i \le n\), then \(H_i = \log_2 (x_i)\).
\end{lemma}
\appendixproof{lem: lowes-on-ent}{

For \(x_i \le n\), we have
\[
H_i = \sum_{j=1}^{x_i} \frac{1}{x_i} \log_2 (x_i) = \log_2 (x_i).
\]

For \(x_i > n\), due to the special construction of the sequence, exist numbers $t$ and $w$ such that there are \(t\) occurrences of \(1, 2, \ldots, n\) that have been accessed \(w + 1\) times, and the remaining \(n - t\) elements have been accessed \(w\) times. Hence, we obtain

\begin{align*}
H_i &= \frac{t(w + 1) \log_2 \left(\frac{nw + t}{w + 1}\right) 
+ (n - t)w \log_2 \left(\frac{nw + t}{w}\right)}{nw + t} \\
&\ge \frac{(nw + t) \log_2 \left(\frac{nw + t}{w + 1}\right)}{nw + t} \\
&= \log_2 \left(\frac{nw + t}{w + 1}\right).
\end{align*}

\(x_i > n\) implies  \(w \ge 1\), it follows that \(\log_2\left(\frac{nw + t}{w + 1}\right) \ge \log_2\left(\frac{n}{2}\right)\). For \(n \ge 4\), we have \(\log_2\left(\frac{n}{2}\right) \ge \log_2 (\sqrt{n})\). Therefore,
\[
H_i \ge \log_2 (\sqrt{n}) = \frac{\log_2 (n)}{2}.
\]
}

\begin{lemma}\appref{lem: low-tree}
The expression \( kC(n) + \sum_{i=1}^{k+1} x_i \min\{\log_2 (x_i), \log_2 (n)\} \) with constraints having $k$ and $x_i$'s as integers and \(\sum_{i=1}^{k+1} x_i = m \)  has \( m \log_2 (C(n))\) as a lower bound.

\end{lemma}

\appendixproof{lem: low-tree}{
Let \( x_{p_i} \leq n \) for \( 1 \leq i \leq t \), and \( x_{p_i} > n \) for \( t+1 \leq i \leq t+j \), where \( (p_i)_{i=1}^n \) is a permutation of \( \{1, \ldots, n\} \). Using the convexity of \( x \log_2 x \) and the fact that \( x \log_2 n \) is linear in \( x \), it follows that:
\[
kC(n) + \sum_{i=1}^{k+1} x_{p_i} \min\{\log_2 (x_{p_i}), \log_2(n)\} \geq (j + t - 1)C(n) + t s \log_2 (s) + j y \log_2 (n)
\]
where \( s \) is the average of \( (x_{p_i})_{i=1}^t\), and \( y \) is the average of \( (x_{p_i})_{i=t+1}^{t+j}\). Given \( ts + jy = m \), we now relax the problem on real numbers and solve the problem using the method of Lagrange multipliers as stated on page 285 of the book \emph{Mathematical Methods and Models for Economists}~\cite{de2000mathematical}
the objective function will be noted as \(
f(j,t,s,y) = -((j + t - 1)C(n) + ts \log_2 (s) + jy \log_2 (n)).
\) And the constraint equality is noted as \(ts + jy = m\).

The partial derivatives are as follows:
\begin{equation*}
\begin{array}{cc}
\frac{\partial }{\partial j} : C(n) + y \log_2 (n) = y \lambda \quad \text{(1)} & 
\frac{\partial }{\partial y} : j \log_2 (n) = j \lambda \quad \text{(2)} \\[10pt]
\frac{\partial }{\partial t} : s \log_2 (s) + C(n) = s \lambda \quad \text{(3)} & 
\frac{\partial }{\partial s} : t \left(\frac{1}{\ln(2)} + \log_2 (s)\right) = t \lambda \quad \text{(4)}
\end{array}
\end{equation*}
we also have: $ts + jy = m \quad \text{(5)}$\\
From equation (2), if \( j \neq 0 \), it implies \( \lambda = \log_2 (n) \). Substituting this into equation (1), we get \( C(n) = 0 \), a contradiction. Therefore, \( j = 0 \), then using equation (5) we have \( t, s \neq 0 \), which ensures constraint qualification and by equation (4), results in \( \frac{1}{\ln(2)} + \log_2 (s) = \lambda \). Substituting into equation (3), we obtain:
\(
s \log_2 (s) + C(n) = s \left(\frac{1}{\ln(2)} + \log_2 (s)\right),
\)
which simplifies to \( C(n) = \frac{s}{\ln(2)} \), or equivalently, \( s = C(n) \ln(2) \). Substituting this result back into equation (5), we have \( t = \frac{m}{C(n) \ln(2)} \), which implies \( k = \frac{m}{C(n) \ln(2)} - 1 \).

Finally, substituting these values back into the original expression gives:
\[
\left(\frac{m}{C(n) \ln(2)} - 1\right)C(n) + \frac{m}{C(n) \ln(2)} \cdot C(n) \ln(2) \log_2 (C(n) \ln(2)) 
\]
\[
=
 \frac{m}{\ln(2)} - C(n) + m \log_2 (C(n) \ln(2)) \ge m \log_2 (C(n))
\]
The last inequality holds for large enough $m$. because $C(n)$ is a function of $n$, without any relation to $m$.
}

Returning to the proof of Lemma~\ref{lem: off-given}, we invoke Theorem~\ref{them: kurt-lower}, restated in Appendix ~\ref{app: threstate} for completeness from the prior work \emph{Nearly Optimal Binary Search Trees}~\cite{mehlhorn1975nearly}. This theorem implies that $c_{\sigma^i} \geq 0.63 \, x_i H_i$, where $H_i$ denotes the entropy of the subsequence served by the tree during stage $i$. Consequently, the total cost admits the following lower bound:
\begin{equation}
\label{eq: tree-based}
\text{Cost} \geq kC(n) + \sum_{i=1}^{k+1} 0.63x_i H_i
.\end{equation}
For the special sequence 
\(
\sigma = \left( \sigma_i \right)_{i=1}^m \quad \text{where} \quad \sigma_i = \left( (i \mod n) + 1 \right)
\)
, using equation ~\ref{eq: tree-based} and Lemma~\ref{lem: lowes-on-ent}, we can now conclude that
\begin{equation}   
\label{eq:cost_1}
\text{Cost} \ge 0.31\left( kC(n) + \sum_{i=1}^{k+1} x_i \min\{\log_2 (x_i), \log_2 (n)\}\right).
\end{equation}
By applying Lemma~\ref{lem: low-tree} to the right-hand side of Equation~\ref{eq:cost_1}, we obtain the desired result.
\end{proof}
\subsection{All-to-all model}

We now turn to the all-to-all communication setting. In Lemma~\ref{lem: up_op_off}, we show that Algorithm~\ref{alg:offline}, when applied by treating each communicating pair as two separate single-source requests, achieves a total cost of at most  
\(
4m \log_2\bigl(C(n)\bigr) + 3.9.
\)  
We also establish a corresponding lower bound of  
\(
\tfrac{1}{4} m \log_2\bigl(C(n)\bigr)
\)
in Lemma~\ref{lem: offoptlow}.
\begin{lemma}\appref{lem: up_op_off}
Exist offline algorithms which at most incur the cost \\$4m(\log_2 (C(n))+ 3.9)$ for any sequence $\sigma = \left( (\sigma_{s_i}, \sigma_{d_i}) \right)_{i=1}^m$.
\end{lemma}
\appendixproof{lem: up_op_off}
{We prove this by providing an algorithm with a cost of at most $4m(\log_2 (C(n))+ 3.9)$. The algorithm is similar to the one described in Lemma~\ref{lem: off-given}. Following the notation introduced in Proposition ~\ref{prop:stage-cost-all-to-all}, the cost is given by 
\(kC(n) + \sum_{i=1}^{k+1} c_{\sigma^i}.\)
where $c_{\sigma^i}$ is the cost of answering the $i_{th}$ subsequence of the sequence $\sigma$ and $x_i$ is it's length.  
Let $k = \lceil \frac{m}{C(n) \ln(2)} \rceil - 1$, and let the $x_i$ values be either $\lceil \frac{m}{\lceil \frac{m}{C(n) \ln(2)} \rceil} \rceil$ or $\lfloor \frac{m}{\lceil \frac{m}{C(n) \ln(2)} \rceil} \rfloor$, chosen so that their sum equals $m$. As shown in the proof of Lemma~\ref{lem: off-given}, all $x_i$ values are bounded by $C(n) \ln(2) + 1$ and $k\le \frac{m}{C(n)\ln(2)}$.
Now, consider each $i$-th subsequence of $\sigma$ as in single-source model with a query sequence having $2x_i$ elements. We use the optimal static binary tree for accessing these $2x_i$ elements in order. The shortest path between nodes $\sigma_{s_j}$ and $\sigma_{d_j}$ in a binary tree is at most the sum of the costs of finding each node individually in the tree, meaning $c_{\sigma^i}$ is bounded by the cost of accessing each node separately.

Using the tree obtained from Theorem~\ref{them: kurt-upper} for each subsequence, we have:
\begin{align*}
\text{Cost} 
&= kC(n) + \sum_{i=1}^{k+1} c_{\sigma^i} \\
&\leq kC(n) + 4m + \sum_{i=1}^{k+1} \frac{2x_i H_i}{1 - \log_2 (\sqrt{5} - 1)} \\
&\leq kC(n) + 4\left(m + \sum_{i=1}^{k+1} x_i H_i\right)
\end{align*}
where $H_i$ is the entropy of the $i$-th subsequence, considered as a queried sequence of nodes. \\
Each stage has $d_i$ distinct elements, where $d_i \leq 2x_i \leq 2C(n) \ln(2) + 2$. Therefore, 
\[
H_i \leq \log_2 (d_i)  \leq \log_2 (2(C(n) \ln(2) + 1)).
\]
We can now write:
\[
\text{Cost} \leq  k C(n) +4( m + \sum_{i=1}^{k+1} x_i \log_2 (2(C(n) \ln(2) + 1)))\le \]\[  \frac{m}{C(n)\ln(2)} C(n) +4( m + m \log_2 (2(C(n) \ln(2) + 1))) \le 4m(\log_2 (C(n))+ 3.9)
\]

}
\begin{lemma}\appref{lem: offoptlow}
Given the sequence 
\[
\sigma = \left( (\sigma_{s_{i}}=1,\sigma_{d_{i}}) \right)_{i=1}^m \quad \text{where} \quad \sigma_{d_{i}} = \left( (i \mod n) + 1 \right)
\]
with \( m \) elements, for any offline algorithm, we have: 
\(
\text{Cost} \ge \frac{m \log_2 (C(n))}{4}.
\)
\end{lemma}

\appendixproof{lem: offoptlow}{To prove Lemma~\ref{lem: offoptlow}, we first state and prove Lemma~\ref{lem: rot-lem}.

\begin{lemma}
\label{lem: rot-lem}
Let \(\sigma\) be a sequence of \(e\) communications, where each communication originates from a fixed source node \(x\), i.e., \(\sigma_{s_i} = x\) for all \(i\). In the optimal static binary search tree constructed for this sequence, at least one of the following conditions holds:  
\begin{enumerate}
    \item The total frequency of the destination nodes within the subtree rooted at \(x\) is at least \(\frac{e}{4}\), or  
    \item If \(y\) denotes the parent of \(x\), then the combined frequency of \(y\) and the subtree rooted at its other child is at least \(\frac{e}{4}\).
\end{enumerate}
And in each case, the incurred cost is at least as much as answering  \(\frac {e}{4}\) elements routed from the root in a binary search tree.
\end{lemma}

\begin{proof}
        
\begin{figure}[h]
    \centering
    \begin{tikzpicture}[scale=1]
        \tikzstyle{treenode} = [circle, draw, minimum size=8mm, inner sep=1pt]

        \node[treenode] (z1) at (3, 3) {z};
        \node[treenode] (y1) at (2, 2) {y};
        \node[treenode] (x1) at (1, 1) {x};

        \draw[fill=gray!40] (0.25,-0.3) -- (-0.2,-1) -- (0.7,-1) -- cycle;
        \draw[fill=gray!40] (1.7,-0.3) -- (1.2,-1) -- (2.2,-1) -- cycle;
        \draw[fill=gray!40] (3.1,-0.3) -- (2.6,-1) -- (3.6,-1) -- cycle;

        \node at (0.25,-0.7) {A};
        \node at (1.7,-0.7) {B};
        \node at (3.1,-0.7) {C};

        \draw (z1) -- (y1);
        \draw (y1) -- (x1);
        \draw (y1) -- (3.1,-0.3);
        \draw (x1) -- (0.25,-0.3);
        \draw (x1) -- (1.7,-0.3);

        \node[treenode] (z2) at (8, 3) {z};
        \node[treenode] (x2) at (7, 2) {x};
        \node[treenode] (y2) at (8, 1) {y};

        \draw[fill=gray!40] (5.3,-0.3) -- (4.8,-1) -- (5.8,-1) -- cycle;
        \draw[fill=gray!40] (6.9,-0.3) -- (6.4,-1) -- (7.4,-1) -- cycle;
        \draw[fill=gray!40] (8.9,-0.3) -- (8.4,-1) -- (9.4,-1) -- cycle;

        \node at (5.3,-0.7) {A};
        \node at (6.9,-0.7) {B};
        \node at (8.9,-0.7) {C};

        \draw (z2) -- (x2);
        \draw (x2) -- (5.3,-0.3);
        \draw (x2) -- (y2);
        \draw (y2) -- (6.9,-0.3);
        \draw (y2) -- (8.9,-0.3);

        \draw[->, thick] (3.5, 1.5) -- (5.5, 1.5) node[midway, above] {Right rotation about $y$};

    \end{tikzpicture}
    \caption{Right rotation in a binary search tree.}
    \label{fig:tree-rotation}
\end{figure}
If the root of the optimal static binary search tree is the source node, then the total frequency of all destination nodes in the subtree rooted at the source node equals \( e \ge \frac{e}{4}\).

If the root of the optimal static binary search tree is not the source node, then \( x \) is either the left or right child of a node \( y \) in the optimal static binary search tree. Performing a left or right rotation about \( y \), accordingly, should not result in a more expensive structure.

We analyze in detail the case where \( x \) is the left child of \( y \), as the other case follows similarly. Figure~\ref{fig:tree-rotation} illustrates this scenario. After the rotation, the access costs for the nodes in set \( A \) (with frequency denoted as \( \text{freq}(A) \)), node \( x \) (with frequency denoted as \( \text{freq}(x) \)), set \( C \) (with frequency denoted as \( \text{freq}(C) \)), and node \( y \) (with frequency denoted as \( \text{freq}(y) \)) remain unchanged. However, the nodes in set \( B \) (with frequency denoted as \( \text{freq}(B) \)) will now be accessed at an incremented cost of one, while the nodes in the remaining parts of the tree (with frequency denoted as \( \text{freq}(\text{Other}) \)) will be accessed at a cost decremented by one. For the tree to remain optimal, it must satisfy the condition
\(
\text{freq}(\text{Other}) \leq \text{freq}(B).
\)
As a result, we have:
\[
\text{freq}(A) + \text{freq}(x) + \text{freq}(C) + \text{freq}(y) + 2\text{freq}(\text{Other}) \le \] \[ \text{freq}(A) + \text{freq}(B) + \text{freq}(x) + \text{freq}(C) + \text{freq}(y) + \text{freq}(\text{Other}) = e.
\]

This leads to:
\begin{gather*}
\text{freq}(A) + \text{freq}(x) + \text{freq}(C) + \text{freq}(B) + \text{freq}(y) 
= \text{freq}(A, x, C, B, y) \\
\geq \frac{e + \text{freq}(A, x, C, y)}{2} 
\geq \frac{e}{2}.
\end{gather*}

As a consequence, we conclude that either
\[
\text{freq}(A) + \text{freq}(x) + \text{freq}(B) \geq \frac{e}{4}
\quad \text{or} \quad
\text{freq}(C) + \text{freq}(y) \geq \frac{e}{4}.
\]

In the first case, we incur a cost associated with accessing at least \( \frac{e}{4} \) elements routed from the source node in the subtree rooted at \( x \). In the second case, we incur a cost associated with accessing at least \( \frac{e}{4} \) elements routed from the source node in the subtree rooted at \( y \), while disregarding the subtree rooted at \( x \).
\end{proof}

We now return to the proof of Lemma~\ref{lem: offoptlow} and proceed as follows.

Following Proposition ~\ref{prop:stage-cost-all-to-all}, the cost is given by 
\(kC(n) + \sum_{i=1}^{k+1} c_{\sigma^i}.\)
where $c_{\sigma^i}$ is the cost of answering the $i_{th}$ subsequence of the sequence $\sigma$ and $x_i$ is it's length.
By applying Lemma~\ref{lem: rot-lem}, it follows that \( c_{\sigma^i} \) is greater than the cost of answering $\frac{x_i}{4}$ elements in an optimal static binary tree designed for a sequence \( i^{\prime} \).
where the sequence $i^{\prime}$ is the $i_{th}$ subsection of the sequence $\sigma$ filtered by only nodes that are in those at least $\frac{e}{4}$ elements as stated in Lemma~\ref{lem: rot-lem}.
This sequence \( i^{\prime} \) contains at least \( \lceil \frac{x_i}{4} \rceil \) elements. Furthermore, there exists a natural number \( w \) such that some elements in \( i^{\prime} \) are repeated \( w \) times, while others are repeated \( w+1 \) times.

If \( x_i \leq n \), the entropy of \( i^{\prime} \) is the logarithm of its cardinality, which is at least \( \log_2 \left(\frac{x_i}{4}\right) \). If \( x_i > n \), or equivalently \(w\ge1\), we will show that the entropy of \( i^{\prime} \) is at least \( \frac{\log_2 \left(\frac{n}{8}\right)}{2} \). To prove this, we first claim that there are at least \( \frac{n}{8} \) distinct elements in \( i^{\prime} \). 

Suppose, for the sake of contradiction, that this is not the case. Then the cardinality of \( i^{\prime} \) would be less than \( \frac{n}{8}(w+1) \). However, by Lemma~\ref{lem: rot-lem}, we know that the cardinality of \( i^{\prime} \) is at least \( \frac{x_i}{4} \), and we also have \( nw \leq x_i \). Combining these results, we get 
\[
\frac{nw}{4} \le \frac{x_i}{4} \leq |i^{\prime}| \leq \frac{n}{8}(w+1),
\]
which implies \( w \leq 1 \). Since \( w \) is a natural number, it must be that \( w = 1 \). However, this leads to a contradiction because it would mean all elements in the \( i \)-th subsequence are repeated \( w \) times, while all elements in \( i^{\prime} \) are repeated \( w+1 \) times, which is impossible.

Now, we find a lower bound for the entropy of \( i^{\prime} \). Let \( s \) elements in \( i^{\prime} \) be repeated \( w \) times, and \( t \) elements be repeated \( w+1 \) times. We have already shown that \( s + t \geq \frac{n}{8} \). Thus, the entropy of \( i^{\prime} \) satisfies:
\[
\text{Entropy}(i^{\prime}) = \frac{s w \log_2 \left(\frac{s w + t (w+1)}{w}\right) + t (w+1) \log_2 \left(\frac{s w + t (w+1)}{w+1}\right)}{s w + t (w+1)} 
\]
\[
\geq \log_2 \left(\frac{s w + t (w+1)}{w+1}\right) \geq \log_2 \left(\frac{s + t}{2}\right) \geq \log_2 \left(\frac{n}{16}\right).
\]
This implies that \(
c_{\sigma^i} \geq \frac{x_i}{4} \min\left\{ \log_2 \left(\frac{x_i}{4}\right), \log_2 \left(\frac{n}{16}\right) \right\}.
\)Using Theorem~\ref{them: kurt-lower}, we can write:
\[
\text{Cost} = kC(n) + \sum_{i=1}^{k+1} c_{\sigma^i} \geq kC(n) + \sum_{i=1}^{k+1} \frac{x_i}{4} \min\left\{ \log_2 \left(\frac{x_i}{4}\right), \log_2 \left(\frac{n}{16}\right) \right\} \]
\[\ge kC(n)+ \sum_{i=1}^{k+1} \frac{x_i}{4} \min\left\{ \log_2 \left(x_i\right), \log_2 \left(n\right) \right\} -(k+1) 
\]
Now, applying Lemma~\ref{lem: low-tree}, we get 
\(
\text{Cost} 
  \ge\frac{m \log_2 (C(n))}{4}.
\)
}
\section{Analysis of online algorithms}
\label{sec: alg}
In this section, we propose a mathematical framework for any online algorithm tailored to single-source or all-to-all adaptive BST models, and present our online algorithm for the single-source case, which can be extended to the all-to-all case.
We then show that no online algorithm can achieve a competitive ratio better than
\[
\frac{\log_2(n)}{2 \log_2(C(n) + 2.2)} \; \text{(single-source)} \; \text{and} \; \tfrac{1}{4} \log_2 n \; \text{(all-to-all)}.
\]
\subsection{A Framework for Online Adaptive BSTs}

To reason about the behavior of online adaptive binary search tree algorithms under general request patterns, we introduce a structural decomposition that applies uniformly to both single-source and all-to-all models. This decomposition abstracts the evolution of the tree into three components: a history summarizer, a restructuring predicate, and a transformation function. The assumption that the request sequence is generated obliviously—independently of the algorithm's internal state—simplifies our reasoning. The following lemma formalizes this decomposition.
\begin{lemma}\appref{lem: decomposition}
Suppose the request sequence \((\sigma_t)_{t \geq 1}\) is generated independently of the state of the binary search tree at each time step; that is, for all \(t\), the generation of \(\sigma_t\) is oblivious to \(T_i\) for all \(i \le t\), and to the internal state of the algorithm.

Requests \(\sigma_t\) may, for all $t$, represent either single-node accesses (e.g., \(\sigma_t = u\)) or communication requests between pairs of nodes (e.g., \(\sigma_t = (u, v)\)).

Under this assumption, and given the cost functions defined in Section~\ref{sec: model}, any deterministic online adaptive binary search tree algorithm can be expressed as a triple of functions \((q, g, I)\), such that the tree at time \(t+1\) evolves according to
\[
T_{t+1} =
\begin{cases}
f(q_t):=g(q_t), & \text{if } I(q_t, T_t) = 1, \\
T_t, & \text{otherwise},
\end{cases}
\]
where:
\(T_t\) denotes the tree at time \(t\),
\(q_t = q(\sigma_1, \dots, \sigma_t)\) is a summary of the request history up to time \(t\),
\(g\) is a deterministic function mapping request summaries to binary search trees, and
\(I\) is a Boolean predicate determining whether a restructuring is applied.

\end{lemma}

\appendixproof{lem: decomposition}{

Let \textsc{ALG} be a deterministic online binary search tree (BST) algorithm. As an online algorithm, it only has access to the tree states \( T_i \) for all \( i \leq t \) and the access history \( (\sigma_1, \dots, \sigma_t) \) at time \( t \). Since it is deterministic, it produces a unique tree \( T_{t+1} \) based solely on this information. Moreover, because the access sequence is generated obliviously---without knowledge of the trees \( T_i \) for \( i \leq t \)---and the cost function defined in Section~\ref{sec: model} depends only on the current tree structure, the earlier tree configurations \( T_i \) for \( i < t \) are redundant given access to \( T_t \).

Define $q(\sigma_1, \dots, \sigma_t) := (\sigma_1, \dots, \sigma_t)$ to store the entire access history (or a sufficient summary thereof). Let $g(q_t, T_t) := ALG(q_t, T_t)$ denote the next tree computed by the algorithm, and define $I(q_t) := 1$ if $ALG(q_t, T_t) \neq T_t$, and $0$ otherwise.

Then the update rule becomes:
\[
T_{t+1} = 
\begin{cases}
g(q_t), & \text{if } I(q_t,T_t) = 1, \\
T_t, & \text{otherwise},
\end{cases}
\]
which exactly reproduces the behavior of ALG at every time step.
}
\begin{algorithm}[h]
\caption{Proposed Online Algorithm with Budget-Based Routing}
\label{alg :prptree}
\SetKwFunction{Routing}{Routing}
\SetKwProg{Fn}{Function}{:}{}
\SetKwInOut{Input}{Input}
\SetKwInOut{Output}{Output}

\Input{$C(n)$: Restructuring cost\\
       $m$: Number of requests\\
       $t_0$: Initial tree\\
       $\sigma_i$: Source node of request $i$}

\BlankLine
$currentCost \gets 0$\;
$tree \gets t_0$\;

\For{$i \gets 1$ \KwTo $m$}{
    \If{$currentCost < C(n)$}{
        $d \gets$ \Routing{$C(n) - currentCost$, $tree$, $\sigma_i$}\;
        \If{$d \neq$ \text{Null}}{
            $currentCost \gets currentCost + d$\;
        }
        \Else{
            $tree \gets g(q(\sigma_{1{:}i}))$\;
            $currentCost \gets$ \Routing{$C(n)$, $tree$, $\sigma_i$}\;
        }
    }
    \Else{
        $tree \gets g(q(\sigma_{1{:}i}))$\;
        $currentCost \gets$ \Routing{$C(n)$, $tree$, $\sigma_i$}\;
    }
}

\BlankLine
\Fn{\Routing{budget, tree, target}}{
    $node \gets$ tree.root\;
    
    \While{budget $> 0$}{
        \If{$node = target$}{
            \KwRet budget\;
        }
        \ElseIf{$target < node$}{
            $node \gets$ node.leftChild\;
        }
        \ElseIf{$target > node$}{
            $node \gets$ node.rightChild\;
        }
        \Else{
            \KwRet Null\;
        }
        budget $\gets$ budget - 1\;
    }
    \KwRet Null\;
}
\end{algorithm}
Using this decomposition, we propose Algorithm~\ref{alg :prptree} for the single-source model, where the predicate \(I(q_t, T_t)\) is computed internally. 

The function \(q\) maps the request prefix \((\sigma_1, \dots, \sigma_t)\) to a predictive summary \(q_t\) that informs future decisions. This summary may represent an explicit model or a sufficient statistic, and can be instantiated using a range of techniques, including n-gram models~\cite{frederick1999statistical} for local frequency estimation, HMMs~\cite{rabiner2002tutorial} for latent-state modeling, AR/ARIMA models~\cite{box2015time} for forecasting, RNNs, LSTMs, and GRUs~\cite{rumelhart1986learning, hochreiter1997long, cho2014learning} for capturing sequential dependencies, Transformers~\cite{vaswani2017attention} for attention-based global modeling, pointer networks~\cite{vinyals2015pointer} for predicting subsets of past requests, and Bayesian nonparametric models such as HDPs~\cite{teh2006hierarchical} for flexible and adaptive modeling. These methods vary in computational and statistical properties; the appropriate choice depends on the request dynamics and tolerance to model error.

Given \(q_t\), the function \(g\) constructs a static BST aimed at minimizing expected access cost. Options include the optimal static BST (computable in \(O(n^2)\) via dynamic programming~\cite{knuth1971optimum}), greedy or heuristic-based constructions~\cite{mehlhorn1975nearly} that are faster but approximate, and randomized structures such as treaps, especially when \(q_t\) yields a probabilistic priority distribution. 

In Appendix~\ref{app: B}, we show that if \(q\) and \(g\) are suitably instantiated, the proposed restructuring rule \(I\) yields a total cost that is, in expectation, at most a constant factor larger than the cost incurred by the online algorithm \((q, g, I^{\mathrm{opt}})\). 

The framework also works for all-to-all settings, where each request $\sigma_t$ is between two nodes. We can extract summaries for each pair, and run the algorithm on those. The function $g$ then builds a tree that tries to keep the communication cost low, based on the predicted demand.

\subsection{Competitive Analysis}
\label{com: kol}
In subsection ~\ref{sub: comp} we gave a definition of competitive ratio. Now, will derive lower bounds for the competitive ratio of any online algorithm in each of the models as follows.

\textbf{Single-source model}
In a single-source model Lemma~\ref{lem: CR_TREE}, establishes a lower bound for competitive ratio of any online algorithm.
\begin{lemma}\appref{lem: CR_TREE}
Given a single-source model, any deterministic online algorithm has a competitive ratio of at least $\frac{\log_2(n)}{2 \log_2(C(n) + 2.2)}$.
\end{lemma}

\appendixproof{lem: CR_TREE}{
Consider an arbitrary deterministic online algorithm \( ALG \).  
Since there are nodes in a binary search that have a depth of at least \( \log_2(n) \), no matter how \( ALG \) modifies the tree over time, there exists a sequence of length \( m \), each element accessing one of the deepest nodes of the temporary tree, for which \( ALG \) incurs a cost of at least \( m \cdot \log_2(n) \). We ignore any additional adjustment cost incurred by \( ALG \).  
By Lemma~\ref{lem: off-given}, the optimal offline algorithm pays at most \( 2m \cdot \log_2(C(n) + 2.2) \) on the same sequence. Thus, the competitive ratio satisfies
\[
\mathrm{CR} \ge \frac{\log_2(n)}{2 \log_2(C(n) + 2.2)}.
\]
}

\textbf{All-to-all model}
To establish a lower bound on the competitive ratio of any online algorithm, we can employ a reasoning approach analogous to that used in the proof of Lemma~\ref{lem: CR_TREE}. However, it leads to a relatively weak lower bound on the competitive ratio. To address this limitation, we introduce additional theoretical tools like Lemma~\ref{lem : best-to-tree} and Lemma~\ref{lem: cons}.
\begin{lemma}\appref{lem : best-to-tree}
Given a binary search tree with $n$ nodes, for any node $x$, the number of nodes greater than $x$ and at a distance less than or equal to $\frac{\log_2 (n)}{2}$ from $x$ is at most $2\sqrt{n} - 2$.
\end{lemma}
\appendixproof{lem : best-to-tree}{
Consider the subtree rooted at the left child of $x$. We can ignore all nodes in this subtree since they are less than $x$. Now, imagine "hanging" the remaining part of the tree from node $x$. The result is a binary tree. In this binary tree, the number of nodes at a depth no greater than $\frac{\log_2 (n)}{2}$ is at most:
\[
\sum_{i=1}^{\lfloor \frac{\log_2 (n)}{2} \rfloor} 2^i = 2^{\lfloor \frac{\log_2 (n)}{2} \rfloor + 1} - 2 \leq 2 \times 2^{\frac{\log_2 (n)}{2}} - 2 = 2\sqrt{n} - 2.
\]
Thus, the number of nodes at a distance of at most $\frac{\log_2 (n)}{2}$ from $x$ is bounded by $2\sqrt{n} - 2$.
}

\begin{figure}[ht]
    \centering
    \begin{minipage}{0.65\textwidth}
        \centering
        \begin{tikzpicture}[level distance=1.5cm, 
                            level 1/.style={sibling distance=3.5cm},
                            level 2/.style={sibling distance=2.5cm},
                            level 3/.style={sibling distance=2cm},
                            every node/.style={circle, draw, minimum size=7mm, inner sep=1pt}]
            \node (root) {\textcolor{red}{1}}
                child [missing] {}
                child {node {\textcolor{red}{4}}
                    child {node {\textcolor{blue}{2}}
                        child [missing] {}
                        child {node {\textcolor{blue}{3}}}
                    }
                    child {node {\textcolor{green!60!black}{5}}
                        child [missing] {}
                        child {node {\textcolor{green!60!black}{6}}}
                    }
                };
        \end{tikzpicture}
    \end{minipage}
    \hfill
    \begin{minipage}{0.3\textwidth}
        \captionof{figure}{A binary search tree in correspondence with the non-self-intersecting matching (1,4), (2,3), (5,6).}
        \label{fig: 3}
    \end{minipage}
\end{figure}
\begin{lemma}\appref{lem: cons}
For any non-self-intersecting matching\footnote{A \textbf{non-self-intersecting matching} for nodes $1$ to $n$ is a matching such that for any two edges, one between nodes $x$ and $y$ and another between nodes $z$ and $w$ (with $x < y$ and $z < w$), the intervals $[x, y]$ and $[z, w]$ either do not overlap or one is a proper subset of the other.} of nodes $1$ to $n$, there exists a corresponding static binary search tree such that the vertices of any two edges in the matching are adjacent in this binary search tree(see Figure~\ref{fig: 3} as an example). 
\end{lemma}

\appendixproof{lem: cons}
{
We proceed by induction on \( n \).

\textbf{Base Case:} When \( n = 2 \), the only possible non empty non-intersecting matching is \( (1,2) \). In this case, the corresponding binary search tree (BST) consists of a root node labeled \( 1 \) with \( 2 \) as its right child. This satisfies the required structure.

\textbf{Inductive Hypothesis:} Suppose that for all values \( n < k \), the claim holds; that is, every non-intersecting matching corresponds to a valid BST.

\textbf{Inductive Step:} We aim to prove the claim for \( n = k \).

Let \( j \) be the smallest number that is matched in the non-intersecting matching, and let \( i \) be the number to which \( j \) is matched. Consider constructing a BST where \( j \) serves as the root. By the structure of non-intersecting matchings, all elements less than \( j \) form an independent sub-matching and can be recursively inserted into the left subtree of \( j \) by the induction hypothesis. The element \( i \) is placed as the right child of \( j \), and all numbers between \( i \) and \( j \) (if any) can be inserted into the left subtree of \( i \) according to the induction hypothesis. Similarly, all elements greater than \( i \) form another independent sub-matching and can be recursively inserted into the right subtree of \( i \) as a valid BST, again by the induction hypothesis.
Thus, by the principle of mathematical induction, the claim holds for all \( n \).
}

\begin{lemma}
Given an all-to-all model, any deterministic online algorithm, has a competitive ratio greater than $\frac{\log_2 (n)}{4}$.
\end{lemma}
\begin{proof}
Following the notion introduced in Proposition~\ref{prop:stage-cost-all-to-all}, we will construct a sequence of length \( m \) such that the online algorithm incurs a cost of at least \( \frac{m \log_2 (n)}{2} + kC(n) \). Meanwhile, the optimal offline algorithm incurs a cost less than \( m + \frac{4kC(n)}{\sqrt{n}} \) for the same sequence.

To construct such a sequence, we follow a greedy strategy:
\begin{itemize}
    \item We start with an empty matching and line up all the nodes 1 to \( n \) in ascending order.
    \item Find the node with the smallest number that is not part of the current matching and is not within any of the intervals formed by the vertices of the current matching edges. Denote this node as \( x \), and mark it as the source\(s\).
    \item Next, select the node with the smallest number greater than \( x \) such that its distance to \( x \) in the binary search tree is greater than \( \frac{\log_2 (n)}{2} \). Mark this node as the destination\(d\).
    \item If no such node can be found, return to the first step.
    \item repeat the communication \((s,d)\) until the online algorithm changes the tree's structure. Then return to the second step.
\end{itemize}

Now, we describe an offline algorithm that pays less than \( m + \frac{4kC(n)}{\sqrt{n}} \) for processing the sequence. Each time the matching reaches its limit, the offline algorithm changes the tree structure at a cost of \( C(n) \), Lemma~\ref{lem: cons} ensure that the offline algorithm answers each communication by a unit cost. By Lemma~\ref{lem : best-to-tree}, we can also find that the matching has at least \( \frac{\sqrt{n}}{4} \) edges before reaching its limit because
\[
\left\lfloor \frac{n}{2\sqrt{n} - 1} \right\rfloor > \frac{n}{2\sqrt{n} - 1} - 1 > \frac{\sqrt{n}}{4}.
\]
Thus, the offline algorithm incurs a total cost of less than \( m + \frac{4kC(n)}{\sqrt{n}} \) for processing the sequence.

Hence, we have:
\[
\text{competitive ratio} > \frac{\frac{m \log_2 (n)}{2} + kC(n)}{m + \frac{4kC(n)}{\sqrt{n}}} > \frac{\log_2 (n)}{4}.
\]
\end{proof}

\section{Conclusion}
\label{sec: conclusion}
We studied adaptive binary search trees with fixed restructuring cost \(C(n)\), introducing a unified framework that captures both single-source and all-to-all models. For each model, we designed offline algorithms with logarithmic cost bounds that closely match our lower bounds for specific request sequences. 

On the online side, we also established lower bounds on the competitive ratio in both models and proposed a set-up for deterministic online algorithms for adaptive BSTs, along with a specific algorithm motivated by the goal of optimality in the single-source setting. This decomposition and this algorithm provide a starting point for further investigation and extends naturally to the all-to-all model.

Future work may focus on improving online performance or exploring broader classes of restructuring rules and cost models.
\section{Acknowledgements}
Part of this work was conducted during my internship at the Max Planck Institute for Informatics. I would like to thank Professor Christoph Lenzen for his supervision and guidance throughout the internship. I am also grateful to Professor Stefan Schmid for introducing me to the model and the broader research area, as well as for his helpful suggestions in shaping the research direction. I also thank Dr. Arash Pourdamghani for his valuable comments on earlier drafts of the manuscript.
\bibliography{bbl}

\begin{thebibliography}{10}

\bibitem{avin2019toward}
Chen Avin and Stefan Schmid.
\newblock Toward demand-aware networking: A theory for self-adjusting networks.
\newblock {\em ACM SIGCOMM Computer Communication Review}, 48(5):31--40, 2019.

\bibitem{bentley1975multidimensional}
Jon~Louis Bentley.
\newblock Multidimensional binary search trees used for associative searching.
\newblock {\em Communications of the ACM}, 18(9):509--517, 1975.

\bibitem{borodin2005online}
Allan Borodin and Ran El-Yaniv.
\newblock {\em Online computation and competitive analysis}.
\newblock cambridge university press, 2005.

\bibitem{box2015time}
George~EP Box, Gwilym~M Jenkins, Gregory~C Reinsel, and Greta~M Ljung.
\newblock {\em Time series analysis: forecasting and control}.
\newblock John Wiley \& Sons, 2015.

\bibitem{chalermsook2015greedy}
Parinya Chalermsook, Mayank Goswami, L{\'a}szl{\'o} Kozma, Kurt Mehlhorn, and Thatchaphol Saranurak.
\newblock Greedy is an almost optimal deque.
\newblock In {\em Workshop on Algorithms and Data Structures}, pages 152--165. Springer, 2015.

\bibitem{cho2014learning}
Kyunghyun Cho, Bart Van~Merri{\"e}nboer, Caglar Gulcehre, Dzmitry Bahdanau, Fethi Bougares, Holger Schwenk, and Yoshua Bengio.
\newblock Learning phrase representations using rnn encoder-decoder for statistical machine translation.
\newblock {\em arXiv preprint arXiv:1406.1078}, 2014.

\bibitem{cole2000dynamic}
Richard Cole.
\newblock On the dynamic finger conjecture for splay trees. part ii: The proof.
\newblock {\em SIAM Journal on Computing}, 30(1):44--85, 2000.

\bibitem{cormen2022introduction}
Thomas~H Cormen, Charles~E Leiserson, Ronald~L Rivest, and Clifford Stein.
\newblock {\em Introduction to algorithms}.
\newblock MIT press, 2022.

\bibitem{de2000mathematical}
Angel De~la Fuente and Angel De~La~Fuente.
\newblock {\em Mathematical methods and models for economists}.
\newblock Cambridge University Press, 2000.

\bibitem{demaine2007dynamic}
Erik~D Demaine, Dion Harmon, John Iacono, and Mihai P~a ˇ~tra{\c{s}}cu.
\newblock Dynamic optimality—almost.
\newblock {\em SIAM Journal on Computing}, 37(1):240--251, 2007.

\bibitem{feder2022lazy}
Evgeniy Feder, Ichha Rathod, Punit Shyamsukha, Robert Sama, Vitaly Aksenov, Iosif Salem, and Stefan Schmid.
\newblock Lazy self-adjusting bounded-degree networks for the matching model.
\newblock In {\em IEEE INFOCOM 2022-IEEE Conference on Computer Communications}, pages 1089--1098. IEEE, 2022.

\bibitem{frederick1999statistical}
Jelinek Frederick.
\newblock Statistical methods for speech recognition, 1999.

\bibitem{galperin1993scapegoat}
Igal Galperin and Ronald~L Rivest.
\newblock Scapegoat trees.
\newblock In {\em Proceedings of the fourth annual ACM-SIAM Symposium on Discrete algorithms}, pages 165--174, 1993.

\bibitem{hochreiter1997long}
Sepp Hochreiter and J{\"u}rgen Schmidhuber.
\newblock Long short-term memory.
\newblock {\em Neural computation}, 9(8):1735--1780, 1997.

\bibitem{knuth1971optimum}
Donald~E. Knuth.
\newblock Optimum binary search trees.
\newblock {\em Acta informatica}, 1(1):14--25, 1971.

\bibitem{mehlhorn1975nearly}
Kurt Mehlhorn.
\newblock Nearly optimal binary search trees.
\newblock {\em Acta Informatica}, 5(4):287--295, 1975.

\bibitem{pourdamghani2023seedtree}
Arash Pourdamghani, Chen Avin, Robert Sama, and Stefan Schmid.
\newblock Seedtree: A dynamically optimal and local self-adjusting tree.
\newblock In {\em IEEE INFOCOM 2023-IEEE Conference on Computer Communications}, pages 1--10. IEEE, 2023.

\bibitem{rabiner2002tutorial}
Lawrence~R Rabiner.
\newblock A tutorial on hidden markov models and selected applications in speech recognition.
\newblock {\em Proceedings of the IEEE}, 77(2):257--286, 2002.

\bibitem{rumelhart1986learning}
David~E Rumelhart, Geoffrey~E Hinton, and Ronald~J Williams.
\newblock Learning representations by back-propagating errors.
\newblock {\em nature}, 323(6088):533--536, 1986.

\bibitem{schmid2015splaynet}
Stefan Schmid, Chen Avin, Christian Scheideler, Michael Borokhovich, Bernhard Haeupler, and Zvi Lotker.
\newblock Splaynet: Towards locally self-adjusting networks.
\newblock {\em IEEE/ACM Transactions on Networking}, 24(3):1421--1433, 2015.

\bibitem{seidel1996randomized}
Raimund Seidel and Cecilia~R Aragon.
\newblock Randomized search trees.
\newblock {\em Algorithmica}, 16(4):464--497, 1996.

\bibitem{slastin2023efficient}
Alexander Slastin, Dan Alistarh, and Vitaly Aksenov.
\newblock Efficient self-adjusting search trees via lazy updates.
\newblock {\em arXiv preprint arXiv:2310.05298}, 2023.

\bibitem{sleator1985self}
Daniel~Dominic Sleator and Robert~Endre Tarjan.
\newblock Self-adjusting binary search trees.
\newblock {\em Journal of the ACM (JACM)}, 32(3):652--686, 1985.

\bibitem{tarjan1985amortized}
Robert~Endre Tarjan.
\newblock Amortized computational complexity.
\newblock {\em SIAM Journal on Algebraic Discrete Methods}, 6(2):306--318, 1985.

\bibitem{teh2006hierarchical}
Yee~Whye Teh, Michael~I Jordan, Matthew~J Beal, and David~M Blei.
\newblock Hierarchical dirichlet processes.
\newblock {\em Journal of the american statistical association}, 101(476):1566--1581, 2006.

\bibitem{vaswani2017attention}
Ashish Vaswani, Noam Shazeer, Niki Parmar, Jakob Uszkoreit, Llion Jones, Aidan~N Gomez, {\L}ukasz Kaiser, and Illia Polosukhin.
\newblock Attention is all you need.
\newblock {\em Advances in neural information processing systems}, 30, 2017.

\bibitem{vinyals2015pointer}
Oriol Vinyals, Meire Fortunato, and Navdeep Jaitly.
\newblock Pointer networks.
\newblock {\em Advances in neural information processing systems}, 28, 2015.

\bibitem{wang2006multi}
Chengwen~Chris Wang.
\newblock {\em Multi-splay trees}, volume~68.
\newblock 2006.

\end{thebibliography}
\newpage

\begin{appendices}

\section{Appendix}

\subsection*{Omitted Proofs}

\label{app: omitted proofs}

In this section, we provide the proofs that were excluded from the main text of the paper.

\appendixproofs

\subsection*{Algorithm~\ref{alg :prptree}}
\label{app: B} 
Suppose $q,g$ as defined in section ~\ref{sec: alg} are designed such that $f:=g(q)$ satisfies the following inequality:
\[
\mathbb{E}[\text{Cost}_{f(\sigma_{1:t})}(\sigma_{t+1:T})] \leq \mathbb{E}[\text{Cost}_p(\sigma_{t+1:T})],
\]\footnote{\( \text{Cost}_{\alpha}(\beta) \) represents the cost of answering the sequence \( \beta \) given the tree configuration \( \alpha \).}
where  \(
p \in \left\{ f(\sigma_{1:z}) \mid z < t \right\}
\)
for each $t$.
Lemma ~\ref{lem: const online} shows that Algorithm~\ref{alg :prptree} with the same $q,g$, will incur at most 3 times of the cost that $I^{opt}, q, g$.

\begin{lemma}
\label{lem: const online}
For any online algorithm \( \text{ALG} \), utilizing same function $q, g$ as algorithm~\ref{alg :prptree}, satisfying the above given conditions, we have:
\[
\mathbb{E}[\text{Cost}_{\text{Algorithm}~\ref{alg :prptree}}(\sigma)] \leq 3 \mathbb{E}[\text{Cost}_{\text{ALG}}(\sigma)]
\]

\end{lemma}
\begin{proof}
We proceed by induction on the parameter \( m \).  

Let \( s_1, s_2, \dots \) denote the time steps at which the algorithm \(\text{ALG}\) modifies the tree's structure, and let \( t_1, t_2, \dots \) denote the corresponding times for Algorithm~\ref{alg :prptree}.  

We consider three possible cases:  

Case 1: If \( s_1 = t_1 \), then since both \(\text{ALG}\) and Algorithm~\ref{alg :prptree} utilize the same function \( f \), which operates deterministically, their behavior is identical. Consequently, both algorithms incur the same cost over this portion of \( \sigma \). The claim then follows by applying the induction hypothesis to the remainder of the sequence.  

Case 2: If \( s_1 < t_1 \), Algorithm~\ref{alg :prptree} incurs a cost of at most \( 2C(n) \) over the first \( t_1 \) elements of the sequence. On the other hand, \(\text{ALG}\) pays at least \( C(n) \) over this segment. If we hypothetically allow \(\text{ALG}\) to reconfigure at \( t_1 \) without incurring the reconfiguration cost, the properties of \( f \) ensure that its overall cost decreases. The claim then follows from the induction hypothesis applied to the remainder of the sequence.  

Case 3: If \( s_1 > t_1 \), let \( j \) be the index such that \( t_j \leq s_1 < t_{j+1} \). Our algorithm incurs at most \( (2j+1)C(n) \), over the first \( s_1 \) elements. If we hypothetically allow \(\text{ALG}\) to reconfigure at times \( t_1, \dots, t_j \) and \( s_1 \) without incurring any reconfiguration cost, its expected cost would decrease, making it to pay at least \( jC(n) \) in expectation over the first \( s_1 \) elements and less for the rest of sequence. Since the inequality \( (2j+1)C(n) \leq 3jC(n) \) holds and $f$ is deterministic, the claim follows, with the remainder of the sequence handled by the induction hypothesis.  
\end{proof}

\subsection*{Background Theorems from Prior Work}
\label{app: threstate}
In this section, we restate Theorem~\ref{them: kurt-lower} and Theorem~\ref{them: kurt-upper} from the prior work Nearly Optimal Binary Search Trees by Mehlhorn~\cite{mehlhorn1975nearly}, for completeness.
\begin{theorem}
\label{them: kurt-lower}
Given a sequence $
\sigma = \left( \sigma_i \right)_{i=1}^T 
$
, for any static binary search tree $T$, we have:
$
\text{Cost} \ge T \frac{H(Y)}{\log_2 3}
$
, where Cost is defined as the sum of path lengths needed to route each $\sigma_i$ from the root node, and $Y$ is the empirical measure of the frequency distribution of $\sigma$, with having $H(Y)$ as its empirical entropy.
\end{theorem}
Theorem~\ref{them: kurt-lower} is used in the proofs of Lemma~\ref{lem: offoptlow} and Lemma~\ref{lem: low-Hole}.
\begin{theorem}
\label{them: kurt-upper}
Given sequence \(
\sigma = \left( \sigma_i \right)_{i=1}^T
\), a static binary search tree can be computed using a balancing argument that has the cost that satisfies: 
\(
Cost \leq T \left( 2 + \frac{H(Y)}{1 - \log_2 (\sqrt{5} - 1)} \right),
\)
where Cost is defined as the sum of shortest path lengths needed to route each $\sigma_i$ from the root node and $Y$ is the empirical measure of the frequency distribution of $\sigma$, and $H(Y)$ is its empirical entropy.
\end{theorem}
Theorem~\ref{them: kurt-upper} is used in the proofs of Lemma~\ref{lem: up_op_off} and Lemma~\ref{lem: off-given}.
\end{appendices}

\end{document}